\documentclass[a4paper]{article}
\usepackage{amssymb}
\usepackage[centertags]{amsmath}
\usepackage{amsthm}
\usepackage{graphicx}
\usepackage{float}

\def\Z{{\cal Z}}
\def\I{{\cal I}}
\def\1{\mathrm{Id}}
\def\Zs{{\check {\cal Z}}}
\def\Ss{{\check S}}
\def\Sigmas{{\check \Sigma}}
\def\Us{{\check U}}

\newcommand{\R}{\mathbb{R}}

\newtheorem{thm}{Theorem}[section]
\newtheorem{cor}[thm]{Corollary}
\newtheorem{lem}[thm]{Lemma}
\newtheorem{prop}[thm]{Proposition}

\numberwithin{equation}{section}

\title{\large{{\bf On the eigenvalues of a central configuration}}}

\author{}

\date{}

\begin{document}
	
	\maketitle

\centerline{Alain Albouy$^1$ and Jiexin Sun$^{1,2}$}
\bigskip
\centerline{$^1$ LTE, Observatoire de Paris, CNRS}
\centerline{77 avenue Denfert-Rochereau}
\centerline{F-75014 Paris}
\bigskip
\centerline{$^2$ Department of Mathematics, Shandong University}
\centerline{27 Shanda Nanlu}
\centerline{250100 Jinan, Shandong, P. R. China}

\bigskip
\bigskip

{\bf Abstract.} The equations of the Newtonian $n$-body problem have a matrix form, where an $n\times n$ matrix depending on the masses and on the mutual distances appears as a factor. The $n$ eigenvalues of this matrix are real and nonnegative. In a motion of relative equilibrium, the configuration, called {\it central}, has constant mutual distances. The matrix is constant. We prove that in a relative equilibrium of 5 bodies the two nontrivial eigenvalues are strictly greater than the three trivial ones. This result improves published inequalities about the central configurations, which belong to two independent lines of research. One starts with Williams in 1938 and concerns constraints on the shape of the configuration. The other concerns the Hessian of the potential and its index, and applies to the linear stability of the self-similar motions and to the possible bifurcations.
We also considerably clarify the very useful identities with which Williams discusses his inequalities.

\bigskip

\section{Introduction}

\bigskip

{\bf The Brehm-Wintner-Conley matrix.} Let us write the Newtonian equations of the $n$-body problem in a matrix form. We introduce the Brehm-Wintner-Conley matrix:
$$\Z=
\begin{pmatrix}
\Sigma_1& -m_1 S_{12}&-m_1 S_{13}&\dots& -m_1S_{1n}\cr -m_2S_{12}&\Sigma_2&-m_2S_{23}&\dots &-m_2S_{2n}\cr -m_3 S_{13}& -m_3 S_{23}&\Sigma_3&\dots& -m_3S_{3n}\cr \dots&\dots&\dots&\dots&\dots
\cr
-m_n S_{1n}&-m_n S_{2n}&-m_n S_{3n}&\dots&\Sigma_n
\end{pmatrix},$$
where $m_i>0$ is the mass of body $i$, $i=1,\dots,n$, whose position is the vector $\mathrm{q}_i=(x_i,y_i)\in\R^2$, where $$S_{ij}=\frac{1}{r_{ij}^3},\quad r_{ij}=\sqrt{(x_i-x_j)^2+(y_i-y_j)^2}.$$
Here we restrict our attention to the planar case. The extension to the spatial case is obvious.
To define the $\Sigma_i$'s and to write the Newtonian equations we introduce three vectors of $\R^n$:
$$
U=(1,\dots,1),\quad X=(x_1,\dots,x_n), \quad Y=(y_1,\dots,y_n).
$$
The $\Sigma_i$'s are uniquely defined by the condition $0=U\Z$,
while the Newtonian equations of motion are
\begin{equation}\label{EQn}\ddot X=-X\Z, \quad\ddot Y=-Y\Z.
\end{equation}
Let us give examples of matrix computation. A fundamental remark is that $\Z\mu$ is symmetric, where
$$\mu=\begin{pmatrix}
m_1& 0& 0&\dots& 0\cr 0&m_2&0&\dots &0\cr 0& 0&m_3&\dots& 0\cr \dots&\dots&\dots&\dots&\dots
\cr
0&0&0&\dots&m_n
\end{pmatrix}.$$
It is easy to see that this symmetry expresses the familiar {\it law of action and reaction}. In particular, it gives $\Z\mu U^t=0$, which implies $\ddot X \mu U^t=\ddot Y \mu U^t=0$. We recognize another familiar law: {\it the acceleration of the center of mass is zero.}

\bigskip

{\bf Central configurations.} A configuration $(X,Y)$ is called central if there is a $\lambda$ such that $(\ddot X,\ddot Y)=-\lambda (X,Y)$, that is
\begin{equation}\label{EQcc}
\lambda X=X\Z, \quad\lambda Y=Y\Z.
\end{equation}
In words, $X$ and $Y$ are two eigenvectors of $\Z$ with same eigenvalue $\lambda$. We will always assume that the configuration $(X,Y)$ is noncollinear, which means that $X$ and $Y$ are independent vectors. Another eigenvector is $U$. Its eigenvalue is zero. If $V=(v_1,\dots,v_n)$, then $V\Z\mu V^t=\sum_{i<j}m_im_jS_{ij}(v_i-v_j)^2$. Consequently $\Z\mu$ is positive semidefinite and $U$ generates the kernel. The eigenvalues of $\Z$ are real and nonnegative. Note that any eigenvector $V_i$ different from $U$ has its center of mass at zero, as proved by multiplying $V_i\Z=\lambda_i V_i$ by $\mu U^t$. Indeed, the eigenspaces are mutually orthogonal for the metric defined by $\mu$. If we also want the orthogonality of the eigenvectors $X$ and $Y$, that is, $X\mu Y^t=0$, it is enough to choose the frame $\mathrm{Oxy}$ along the principal axes of inertia of the configuration.

\bigskip

{\bf Relative equilibria.} If $(X_0,Y_0)$ satisfies (\ref{EQcc}), and if $\omega=\sqrt{\lambda}$, then the motion $t\mapsto (X_0\cos\omega t,Y_0\sin\omega t)$ satisfies (\ref{EQn}). It is a motion of relative equilibrium.

\bigskip

{\bf  Results.} The first reasonings about central configurations involving the matrix $\Z$ are due to Brehm \cite{Bre} in 1908. See his equation (14). The comments by Meyer \cite{Mey} (1933) and Wintner \cite{Win} (1941) are reported in Albouy-Fernandes \cite{AF}. The first result about the nontrivial eigenvalues is due to Conley (see \cite{Pac}): in a collinear central configuration, these eigenvalues are strictly greater than the simple eigenvalue $\lambda$. Moeckel \cite{Moe} obtained a second result by extending a statement due to Pacella \cite{Pac}: the arithmetic mean of the nonzero eigenvalues of $\Z$ is greater than $\lambda$. This inequality is true whatever the dimension of the central configuration, and is strict except if the configuration is a regular simplex (see  \cite{AF}). Moeckel \cite{Moe} also proved that Conley's statement does not extend to all the central configurations, by showing a planar example with 474 bodies with an eigenvalue smaller than $\lambda$. The present work is devoted to proving:

\bigskip

\begin{thm} \label{mthm}
In a planar noncollinear central configuration of 5 bodies, 0 is a simple eigenvalue of $\Z$, the multiplier $\lambda>0$ is a double eigenvalue. The two other eigenvalues are strictly greater than $\lambda$.
\end{thm}

This result was expected and conjectured (see \cite{Saa80}, \cite{Moe}, \cite{AF}). What was not expected is that it also solves a conjecture left by Williams \cite{Wil} in 1938. Williams proved relations between the areas of the triangles formed by 3 of the 5 bodies and the quantities $(S_{ik}-\lambda/M)(S_{jl}-\lambda/M)-(S_{il}-\lambda/M)(S_{jk}-\lambda/M)$ where
\begin{equation}\label{EQM}M=m_1+\cdots+m_5.
\end{equation}
The obvious relation of these quantities with the $2\times 2$ minors of $\Z$ has remained unnoticed. Williams then tried to deduce the signs of these quantities. However, his arguments rely on wrong postulates from which he also deduces the impossibility of certain shapes for a central configuration, which are indeed possible\footnote{In \cite{Wil} possible central configurations are wrongly discarded at top of page 566, then, later, through equation (16), which according to \cite{CH} has exceptions, then in Theorem 10.2, in section 12 and in example 2, page 579.  See \cite{Lon} about a related mistake in the 4-body problem.}. Chen and Hsiao \cite{CH} revised his study and proved his main conclusions for the strictly convex configurations. Our result gives all the signs in all the remaining cases. These signs agree with Williams' claims.

We propose to call the eigenvalues of the matrix $\Z$ at a central configuration the {\it eigenvalues of the central configuration}. The vertical oscillations near a self-similar motion are determined by the eigenvalues of the central configuration (see \cite{HOT}).

\bigskip

\section{The areas of the triangles in terms of covectors}

\bigskip

\begin{lem}\label{lem1}
 Consider five points $\mathrm{q}_i=(x_i,y_i)\in\R^2$, $i=1,\dots,5$, not on a same line. Let $\Delta_{hij}=(\mathrm{q}_i-\mathrm{q}_h)\wedge (\mathrm{q}_j-\mathrm{q}_h)$ be twice the oriented area of the triangle $\mathrm{q}_h\mathrm{q}_i\mathrm{q}_j$. Let $\Phi=(\Phi_1,\dots,\Phi_5)\in\R^5$, $\Psi=(\Psi_1,\dots,\Psi_5)\in\R^5$ be independent and satisfy
\begin{equation}\label{EQo}
\sum_{i=1}^5 \Phi_i=\sum_{i=1}^5 \Phi_i x_i=\sum_{i=1}^5 \Phi_i y_i=0,\qquad \sum_{i=1}^5 \Psi_i=\sum_{i=1}^5 \Psi_i x_i=\sum_{i=1}^5 \Psi_i y_i=0.
\end{equation}
Then there is a nonzero $a\in\R$ such that for any even permutation $(1,2,3,4,5)\mapsto(h,i,j,k,l)$
$\Delta_{hij}=a(\Phi_k\Psi_l-\Phi_l\Psi_k)$.

\end{lem}

\bigskip

{\bf Notation.} The sums in (\ref{EQo}) will be considered as duality brackets between the ``vectors''
\begin{equation}\label{EQa}
U=(1,1,1,1,1),\quad X=(x_1,x_2,x_3,x_4,x_5), \quad Y=(y_1,y_2,y_3,y_4,y_5)
\end{equation}
and the ``covectors'' $\Phi$ and $\Psi$.  Using the duality bracket (\ref{EQo}) is written $\langle \Phi,U\rangle=\langle \Phi,X\rangle=\langle \Phi,Y\rangle=0$, etc. We do not want to specify an inner product at this point. We denote the five vectors of the standard base by
$$e_1=(1,0,0,0,0),\quad e_2=(0,1,0,0,0), \quad\dots,\quad e_5=(0,0,0,0,1).$$ The dual base is formed by the five covectors
$$e_1^*=(1,0,0,0,0),\quad e_2^*=(0,1,0,0,0), \quad\dots,\quad e_5^*=(0,0,0,0,1).$$
We also introduce this alternative notation for  twice the oriented area of a triangle
\begin{equation}\label{EQAA}
\Delta^{kl}=\Delta_{hij}\quad\hbox{for any even permutation }(1,2,3,4,5)\mapsto(h,i,j,k,l).
\end{equation}
With this notation, the identity to prove is $\Delta^{kl}=a(\Phi_k\Psi_l-\Phi_l\Psi_k)$.

\bigskip

\begin{proof} We have $$U\wedge X\wedge Y=\sum_{h<i<j}\Delta_{hij}e_h\wedge e_i\wedge e_j.$$
Then
\begin{equation}\label{EQb}
Y\wedge X\wedge U\rfloor e_1^*\wedge e_2^*\wedge e_3^*\wedge e_4^*\wedge e_5^*=\sum_{k<l}\Delta^{kl} e_k^*\wedge e_l^*.\end{equation}
Here $\rfloor$ is the interior product. The interior product of $Y\wedge X\wedge U$ is also the interior product of $U$, followed by the interior product of $X$, followed by the interior product of $Y$.  Let $V$ and $W$ be vectors satisfying $\langle \Phi,V\rangle=1$, $\langle \Phi,W\rangle=0$, $\langle \Psi,V\rangle=0$, $\langle \Psi,W\rangle=1$. Then $(U,X,Y,V,W)$ is a base. The dual base $(U^*,X^*,Y^*,V^*,W^*)$ satisfies $V^*=\Phi$ and $W^*=\Psi$. Then (\ref{EQb}) is
$$Y\wedge X\wedge U\rfloor U^*\wedge X^*\wedge Y^*\wedge V^*\wedge W^*=\Phi\wedge\Psi$$
multiplied by a real factor $a$. But 
$\Phi\wedge\Psi=\sum_{k<l} (\Phi_k\Psi_l-\Phi_l\Psi_k)e_k^*\wedge e_l^*$. 

\end{proof}

\bigskip

{\bf An expression of the real factor in Lemma \ref{lem1}.} We assume that $\R^5$ is endowed with a Euclidean form, not necessarily the standard one. The duality bracket becomes an inner product after identifying vectors and covectors through the Euclidean form. We assume that the covectors $\Phi$ and $\Psi$ satisfy $\|\Phi\|=\|\Psi\|=1$ and $\langle \Phi,\Psi\rangle=0$. Then the factor $a$ in Lemma \ref{lem1} is:
\begin{equation}\label{EQab}a=\pm{ \|U\wedge X\wedge Y\|}\cdot{\|e_1^*\wedge e_2^*\wedge e_3^*\wedge e_4^*\wedge e_5^*\|}.
\end{equation}

\begin{proof} The vectors associated through the Euclidean form to the covectors $\Phi$ and $\Psi$ satisfy the conditions for $V$ and $W$ in the previous proof. Furthermore they are orthogonal unit vectors, also orthogonal to $U$, $X$ and $Y$. We have
$$\|U\wedge X\wedge Y\wedge V\wedge W\|^2=\|U\wedge X\wedge Y\|^2\|V\|^2 \|W\|^2=\|U\wedge X\wedge Y\|^2,$$
$$\|U^*\wedge X^*\wedge Y^*\wedge V^*\wedge W^*\|^2=\|U\wedge X\wedge Y\|^{-2}.$$
But $a U^*\wedge X^*\wedge Y^*\wedge V^*\wedge W^*=e_1^*\wedge e_2^*\wedge e_3^*\wedge e_4^*\wedge e_5^*$ so we get the formula for $a$. 

\end{proof}

\bigskip

{\bf Case of the mass inner product.} In what follows the Euclidean inner product will be defined by the masses. Choosing the orthonormal frame of $\mathrm{Oxy}$ which is centered at the center of mass of the configuration and directed toward the principal axes, we will have $\langle U,X\rangle=\sum m_i x_i=0$, $\langle U,Y\rangle=\sum m_i y_i=0$, $\langle X,Y\rangle=\sum m_i x_iy_i=0$ and consequently
\begin{equation}\label{EQw}
\|U\wedge X\wedge Y\|^2=\|U\|^2\|X\|^2\|Y\|^2=M\Bigl(\sum_{i=1}^5 m_i x_i^2\Bigr)\Bigl(\sum_{i=1}^5 m_i y_i^2\Bigr).
\end{equation}

\bigskip

\section{Shifted BWC matrix and minor determinants}

\bigskip

Although the material of this section easily passes to other numbers of bodies and other dimensions of the configuration, we continue with the restrictions $n=5$ and dimension 2.

\bigskip
{\bf Translation.} The matrix
$$\I=\frac{1}{M}\begin{pmatrix}M-m_1& -m_1&-m_1&-m_1& -m_1\cr -m_2&M-m_2&-m_2&-m_2&-m_2\cr -m_3& -m_3&M-m_3& -m_3& -m_3\cr -m_4 &-m_4 &-m_4 &M-m_4&-m_4 \cr
-m_5 &-m_5 &-m_5 &-m_5 &M-m_5\end{pmatrix}$$
is such that $X\I=(x_1-x_G,x_2-x_G,x_3-x_G,x_4-x_G,x_5-x_G)$ where $x_G$ is defined by $Mx_G=m_1x_1+\cdots +m_5x_5$. If we replace the definition (\ref{EQcc}) of the central configurations by
\begin{equation}\label{EQcct}
\lambda X\I=X\Z, \quad\lambda Y\I=Y\Z,
\end{equation}
then a translated central configuration is now also a central configuration. The center of mass of a configuration solution of (\ref{EQcc}) with $\lambda\neq 0$ should be at the origin, while this is not true for a solution of (\ref{EQcct}). Except for this, the solutions are the same.

\bigskip

{\bf Shifted BWC matrix.} With the above translation choice the central configurations are now defined by the equations
\begin{equation}\label{EQccs}
0=X\Zs, \quad 0=Y\Zs,
\end{equation}
where $\Zs=\Z-\lambda\I$ is the {\it shifted Brehm-Wintner-Conley matrix}, which also satisfies the condition
\begin{equation}\label{EQk}
0=U\Zs.
\end{equation}

\bigskip

{\bf Normalization and amended force.} We normalize the size of a central configuration by the condition $\lambda=M$, where $\lambda$ and $M$ are defined in (\ref{EQcc}) and  (\ref{EQM}). This gives central configurations which are not rescaled when the masses are rescaled. The shifted BWC matrix becomes
$$\Zs=\begin{pmatrix}\Sigmas_1& -m_1 \Ss_{12}&-m_1 \Ss_{13}&-m_1 \Ss_{14}& -m_1\Ss_{15}\cr -m_2\Ss_{12}&\Sigmas_2&-m_2\Ss_{23}&-m_2\Ss_{24} &-m_2\Ss_{25}\cr -m_3 \Ss_{13}& -m_3 \Ss_{23}&\Sigmas_3& -m_3\Ss_{34}& -m_3\Ss_{35}\cr -m_4 \Ss_{14}&-m_4 \Ss_{24}&-m_4 \Ss_{34}&\Sigmas_4&-m_4 \Ss_{45}\cr
-m_5 \Ss_{15}&-m_5 \Ss_{25}&-m_5 \Ss_{35}&-m_5 \Ss_{45}&\Sigmas_5\end{pmatrix},$$
with\begin{equation}\label{EQs}
\Ss_{ij}=S_{ij}-1=\frac{1}{r_{ij}^3}-1,
\end{equation}
and where 
 $\Sigmas_1,\dots,\Sigmas_5$ are characterized by the condition (\ref{EQk}). Equations (\ref{EQccs}) are linear in the masses.

Changing of the BWC matrix $\Z$ into the ``shifted BWC matrix'' $\Zs$ may now be interpreted as changing the law of attraction, since we change the coefficient $1/r_{ij}^3$ of the central force into $1/r_{ij}^3-1$. This new coefficient corresponds to what may be called the amended Newtonian force. The Newtonian potential (or force function)
$$U=\sum_{i<j} \frac{m_im_j}{r_{ij}}$$
is changed into the ``amended Newtonian potential''
$$\Us=\sum_{i<j} m_im_j\Bigl(\frac{1}{r_{ij}}+\frac{r_{ij}^2}{2}\Bigr).$$ We look for the equilibria of this amended force. An equilibrium is a relative equilibrium which has a zero angular velocity.

Let $$\mu=\begin{pmatrix}m_1& 0& 0& 0& 0\cr 0&m_2&0& 0& 0\cr 0& 0&m_3& 0& 0\cr 0& 0& 0&m_4&0\cr
	0& 0& 0& 0&m_5\end{pmatrix}.$$
The matrix $\Zs\mu$ is symmetric of rank at most 2. Since the masses are positive, $\mu$ is positive and there is a common diagonalization
\begin{equation}\label{EQc}
\Zs\mu=\nu_1\Phi\otimes \Phi+\nu_2\Psi\otimes\Psi
\end{equation}
where the covectors $\Phi$ and $\Psi$ satisfy (\ref{EQo}) and $\|\Phi\|^2=\Phi\mu^{-1}\Phi^t=1$, $\|\Psi\|^2=\Psi\mu^{-1}\Psi^t=1$, $\langle \Phi,\Psi\rangle=\Phi\mu^{-1}\Psi^t=0$. The eigenvalues of $\Zs$ are $(0,0,0,\nu_1,\nu_2)$. The eigenvalues of $\Z=\Zs+\lambda\I$ are $(0,\lambda,\lambda,\lambda+\nu_1,\lambda+\nu_2)$, since $\I=\1-(m_1,\dots,m_5)\otimes (1,\dots,1)/M$.

The computation of the $2\times 2$ minors of $\Zs\mu$ by using (\ref{EQc}) is easy. They all have $\nu_1\nu_2$ in factor:
$$(\nu_1\Phi_i\Phi_k+\nu_2\Psi_i\Psi_k)(\nu_1\Phi_j\Phi_l+\nu_2\Psi_j\Psi_l)-
(\nu_1\Phi_i\Phi_l+\nu_2\Psi_i\Psi_l)(\nu_1\Phi_j\Phi_k+\nu_2\Psi_j\Psi_k)$$
$$=\nu_1\nu_2(\Phi_i\Psi_j-\Phi_j\Psi_i)(\Phi_k\Psi_l-\Phi_l\Psi_k).$$
Here we assume $i\neq j$ and $k\neq l$. In the case of minors without a diagonal entry ($i\neq k$, $j\neq l$, $i\neq l$, $j\neq k$) Williams' quantities appear:
$$m_im_km_jm_l(\Ss_{ik}\Ss_{jl}-\Ss_{il}\Ss_{jk})=\nu_1\nu_2(\Phi_i\Psi_j-\Phi_j\Psi_i)(\Phi_k\Psi_l-\Phi_l\Psi_k).$$
Lemma \ref{lem1} transforms this equation into
$$m_im_km_jm_l(\Ss_{ik}\Ss_{jl}-\Ss_{il}\Ss_{jk})a^2=\nu_1\nu_2\Delta^{ij}\Delta^{kl}.$$
The factor $a^2$ is given by (\ref{EQab}) where $$\|e_1^*\wedge e_2^*\wedge e_3^*\wedge e_4^*\wedge e_5^*\|^2=\|e_1^*\|^2\|e_2^*\|^2\|e_3^*\|^2\|e_4^*\|^2\|e_5^*\|^2=(m_1m_2m_3m_4m_5)^{-1}.$$
If we denote by $h$ the missing index, $1\leq h\leq 5$, $h\neq i,j,k,l$, and pass to the other notation of the areas, which does not change the sign whatever the parity of the permutation $12345\mapsto ijhkl$, we get 
\begin{equation}\label{EQW}\Ss_{ik}\Ss_{jl}-\Ss_{il}\Ss_{jk}=\nu m_h\Delta_{ijh}\Delta_{klh},
\end{equation}
where 
\begin{equation}\label{EQX}
\nu=\frac{\nu_1\nu_2}{\|U\wedge X\wedge Y\|^2}.
\end{equation}
The denominator of $\nu$ may be expressed by (\ref{EQw}). Equations (6), (7), (8), (12), (14) of Williams' paper are straightforward consequences of the simpler identities (\ref{EQW}). 

\bigskip

\section{Another approach to identities (\ref{EQW})}

\bigskip

Let us consider the first two vector equations for central configurations, with the normalization and notation (\ref{EQs}): 
\begin{equation}\label{EQAg}
0=\sum_{j\neq 1}m_j\Ss_{1j}(\mathrm{q}_j-\mathrm{q}_1),\qquad 0=\sum_{j\neq 2}m_j\Ss_{2j}(\mathrm{q}_j-\mathrm{q}_2).
\end{equation}
We take the wedge product with $\mathrm{q}_2-\mathrm{q}_1$, and use the notation of Lemma \ref{lem1} for the areas of triangles:
\begin{equation}\label{EQA}
0=\sum_{j>2}m_j\Ss_{1j}\Delta_{12j},\qquad 0=\sum_{j> 2}m_j\Ss_{2j}\Delta_{12j}.
\end{equation}
These are Williams' equations of the first kind (equations (5) of \cite{Wil}). Note that both equations (\ref{EQA}) only differ by an index. If we subtract them we find
\begin{equation}\label{EQAn}
\sum_{j>2}m_j(S_{1j}-S_{2j})\Delta_{12j}=0,
\end{equation}
which is a classical equation which appeared in 1892 in a work by Krediet \cite{Kre}, reappeared in 1905 in a work by Laura \cite{Lau} and in 1906 in a work by Andoyer \cite{And}.

The present work is about $n=5$, and we will show how to deal with (\ref{EQA}) in this case. It is worth mentioning that applying the same reasoning to $n=4$ naturally leads to identities discovered by Dziobek \cite{Dzi} in 1900.

For a planar noncollinear configuration of $n=5$ distinct points, with nonzero masses, (\ref{EQA}) expresses that in $\mathbb{R}^3$ the nonzero vector
\begin{equation}\label{EQB}
(m_3\Delta_{123},m_4\Delta_{124},m_5\Delta_{125})
\end{equation}
is orthogonal to both vectors $(\Ss_{13},\Ss_{14},\Ss_{15})$ and $(\Ss_{23},\Ss_{24},\Ss_{25})$, so, the vector
\begin{equation}\label{EQC}
(\Ss_{14}\Ss_{25}-\Ss_{15}\Ss_{24},\Ss_{15}\Ss_{23}-\Ss_{13}\Ss_{25},\Ss_{13}\Ss_{24}-\Ss_{14}\Ss_{23}),
\end{equation}
is proportional to (\ref{EQB}).
This is a second kind of Williams' equation (equations (6) of \cite{Wil}).
Consider now the 60 fractions:
\begin{equation}\label{EQD}
\nu_{ijhkl}=\frac{\Ss_{ik}\Ss_{jl}-\Ss_{il}\Ss_{jk}}{m_h\Delta_{ijh}\Delta_{klh}},\end{equation}
where $(1, 2, 3, 4, 5)\mapsto (i,j,h,k,l)$ is any even permutation. As (\ref{EQB}) and consequently (\ref{EQC}) may be written for any $(i,j)$ instead of $(1,2)$,
\begin{equation}\label{EQF}
\nu_{ijhkl}=\nu_{ijlhk}
\end{equation} 
for any choice of distinct indices. In words, a circular permutation of the {\it last} three indices of $\nu_{ijhkl}$ does not change the fraction. But (\ref{EQD}) gives 
\begin{equation}\label{EQG}\nu_{ijhkl}=\nu_{klhij}
\end{equation}
which combined with (\ref{EQF}) gives
\begin{equation}\label{EQH}
\nu_{ijhkl}=\nu_{hijkl}
\end{equation}
In words, a circular permutation of the {\it first} three indices of $\nu_{ijhkl}$ does not change the fraction. We also have $\nu_{ijhkl}=\nu_{jihkl}$. We chose to consider only the even permutations of the indices, so we only retain
\begin{equation}\label{EQI}
\nu_{ijhkl}=\nu_{jihlk}.
\end{equation}
Note that in (\ref{EQF}) or (\ref{EQH}), the equality is weaker than in (\ref{EQG}) and (\ref{EQI}). In (\ref{EQG}) and in (\ref{EQI}), the numerators are equal and the denominators are equal. If one fraction is indeterminate, so is the other.

It is easy to see that these relations (\ref{EQF}), (\ref{EQG}), (\ref{EQH}), (\ref{EQI}) generate all the even permutations of the five indices, and that indeed (\ref{EQF}) and (\ref{EQH}) are enough. But this argument is not sufficient to prove that all the fractions (\ref{EQD}) are equal, since if a fraction $b$ is indeterminate, $a=b$ and $b=c$ do not imply $a=c$. Indeterminacies do occur when 3 points are collinear. So, we would have to examine these special cases, which would make this second approach to (\ref{EQW}) slightly longer than the first. The cases with 3 collinear bodies are of interest (see e.g.\ \cite{ChC}).

The following representation may be useful when discussing precisely the transitivity of the permutations (\ref{EQF}) and (\ref{EQH}). Let us fix $(i,j)$ and represent the 3 equalities (\ref{EQF}) as the 3 edges of a triangle, whose vertices represent the 3 fractions. For example with the choice $(i,j)=(1,2)$, applying 3 times (\ref{EQF}) closes a triangle formed by the sequence of fractions with indices
$$12345, 12534, 12453, 12345.$$ Interestingly, extending this representation to the other fractions naturally leads to draw an icosidodecahedron (Figure \ref{DaVinci}). To see how, let us follow a path which continues an edge of the triangle without following the triangle. Applying three times the sequence of permutations (\ref{EQF}), (\ref{EQH}) we find
\begin{equation}\label{EQJ}
12345, 12534, 51234, 51423, 45123, 45312.
\end{equation}
We arrive at the image of 12345 by (\ref{EQG}).
If instead we apply three times the sequence (\ref{EQF}), (\ref{EQH})$^{-1}$ we find
\begin{equation}\label{EQK}
12345, 12534, 25134, 25413, 54213, 54321.
\end{equation}
Here we arrive at the image of 12345 by the sequence (\ref{EQI}), (\ref{EQG}). We decide that a fraction is identified to its image by (\ref{EQG}), so the 60 fractions (\ref{EQD}) are now 30 classes. Each class corresponds to a vertex of the icosidodecahedron. The sequences of type (\ref{EQJ}) close the pentagons.  The sequences of type (\ref{EQK}) join a vertex to the opposite vertex. All these sequences are consistent with the geometry of the icosidodecahedron.

 \begin{figure}[H] 
     \centering     \includegraphics[width=0.35\textwidth, angle=0]{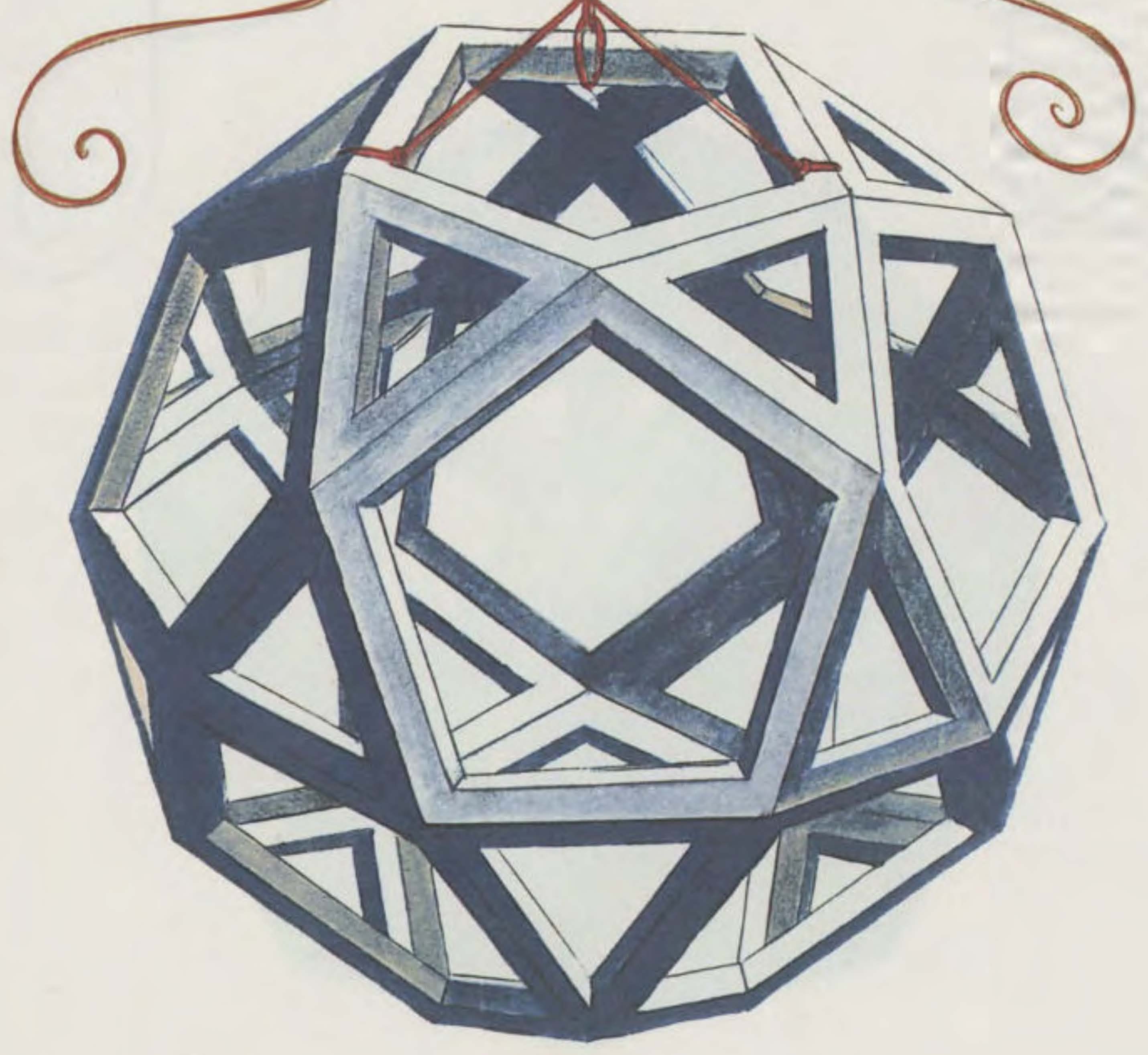}
     \caption{ Icosidodecahedron drawn in {\it Divina proportione} \cite{Paci}.}
     \label{DaVinci}
 \end{figure}

\section{Some inequalities for a 5-body central configuration}

\bigskip

For any noncollinear planar 5-body configuration, the convex hull of the 5 points has the shape of a triangle, a quadrilateral or a pentagon.
We call the ones with triangular or quadrilateral convex hull concave configurations, the ones with pentagonal convex hull strictly convex configurations. We recall that an {\it extreme point} is a point of a convex set which is not on a segment drawn between two other points of the convex set. If the convex set is the convex hull of a configuration of $n$ points, the extreme points are among the $n$ points. We may call them the extreme points of the configuration.  A configuration is called strictly convex if all the points are extreme points. A configuration is convex if all the points are on the boundary of their convex hull. A configuration with triangular convex hull has 3 extreme points, etc. Chen and Hsiao \cite{CH} (2017) prove the following results.

\begin{prop}\cite{CH} 
For a strictly convex 5-body planar central configuration, without loss of generality, we number counterclockwise the bodies as $\mathrm{q}_1$, $\mathrm{q}_2$, $\mathrm{q}_3$, $\mathrm{q}_4$, $\mathrm{q}_5$, then 
 \begin{enumerate}
  \item $r_{13}>r_{12}$, $r_{23}$;\ \ $r_{24}>r_{23}$, $r_{34}$;\ \ $r_{35}>r_{34}$, $r_{45}$;\ \ $r_{14}>r_{45}$, $r_{15}$;\ \ $r_{25}>r_{15}$, $r_{12}$.
  \item $r_{12}>r_{23}\ \Rightarrow r_{14}>r_{34}$;\ \ $r_{23}>r_{34}\ \Rightarrow r_{25}>r_{45}$;\ \ $r_{34}>r_{45}\ \Rightarrow r_{13}>r_{15}$;\ \ $r_{45}>r_{15}\ \Rightarrow r_{24}>r_{12}$;\ \ $r_{15}>r_{12}\ \Rightarrow r_{35}>r_{23}$.
\end{enumerate}
    
\end{prop}

\begin{thm}\label{cov}\cite{CH} 
    For a strictly convex 5-body central configuration, we number the bodies in cyclic order from an arbitrary body as $\mathrm{q}_1$, $\mathrm{q}_2$, $\mathrm{q}_3$, $\mathrm{q}_4$, $\mathrm{q}_5$, then 
    $$\Ss_{12}\Ss_{34}\ >\ \Ss_{13}\Ss_{24}\ >\ \Ss_{14}\Ss_{23}.$$
\end{thm}

We will prove similar results about the mutual distances in concave 5-body central configurations. We begin with an easy but useful result which does not require the configuration to be central and may consequently be applied to a subconfiguration of a central configuration.

\begin{prop}\label{exter} Consider a configuration with four distinct points $\mathrm{q}_1$, $\mathrm{q}_2$, $\mathrm{q}_3$, $\mathrm{q}_4$. If $$r_{12}\leq r_{14}\quad \hbox{and} \quad r_{13}\leq r_{14}$$
then $\mathrm{q}_4$ is strictly exterior to the triangle $\mathrm{q}_1\mathrm{q}_2\mathrm{q}_3$.
\end{prop}

The proof is obvious. The proposition can easily be extended by replacing the triangle by another configuration, $\mathrm{q}_2$ and $\mathrm{q}_3$ by the extreme points and $\mathrm{q}_1$ by any point.

 \begin{figure}[H] 
     \centering     \includegraphics[width=0.3\textwidth, angle=0]{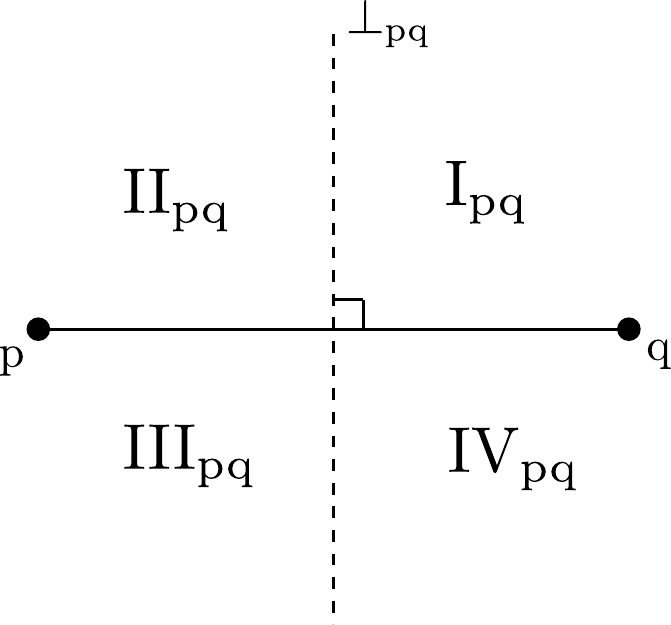}
     \caption{Four quadrants in the plane.}
     \label{PBT}
 \end{figure}

We continue with two results which apply to central configurations of the $n$-body problem. For any 2 points $\mathrm{p}$, $\mathrm{q}$ in the plane, we denote by $\overline{\mathrm{pq}}$ the line segment with the endpoints $\mathrm{p}$ and $\mathrm{q}$, by $\mathrm{pq}$ the line passing through points $\mathrm{p}$ and $\mathrm{q}$ and by $\perp_{\mathrm{pq}}$ the perpendicular bisector of $\overline{\mathrm{pq}}$. We regard $\mathrm{pq}$ as the $\mathrm{x}$-axis and $\perp_{\mathrm{pq}}$ as the $\mathrm{y}$-axis in the plane they are in, then we name the four quadrants of this Cartesian coordinate system counterclockwise starting from the upper right quadrant as $\mathrm{I}_{pq}$, $\mathrm{II}_{pq}$, $\mathrm{III}_{pq}$, $\mathrm{IV}_{pq}$ (see Figure \ref{PBT}). By quadrant we always mean the closed set. We use $\mathrm{X}^{\circ}$ to denote the interior of a closed set $\mathrm{X}$.

\begin{prop}\label{PBTtm}(Perpendicular bisector theorem, Conley)
For any planar central configuration, any two bodies $\mathrm{p}$, $\mathrm{q}$ of the configuration, consider the four open domains $\mathrm{I}_{pq}^{\circ}$, $\mathrm{II}_{pq}^{\circ}$, $\mathrm{III}_{pq}^{\circ}$, $\mathrm{IV}_{pq}^{\circ}$ delimited by the line passing the bodies and by their perpendicular bisector (see Figure \ref{PBT}). There are no bodies in $\mathrm{I}_{pq}^{\circ}\bigcup \mathrm{III}_{pq}^{\circ}$ if and only if there are no bodies in $\mathrm{II}_{pq}^{\circ}\bigcup \mathrm{IV}_{pq}^{\circ}$ (see Figure \ref{PBT2}).
\end{prop}

 \begin{figure}[H] 
     \centering     \includegraphics[width=0.3\textwidth, angle=0]{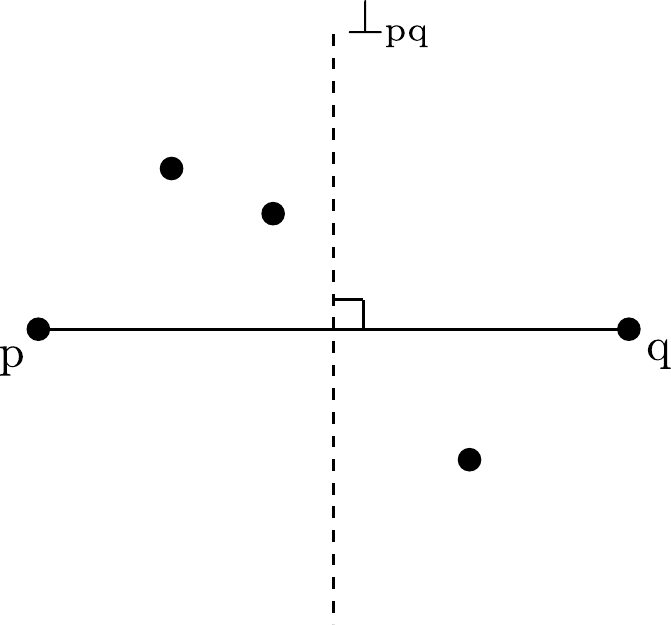}
     \caption{5-body configuration which is not central whatever the masses.}
     \label{PBT2}
 \end{figure}

\begin{proof} Consider the sign of each term of the Krediet-Laura-Andoyer equation (\ref{EQAn}) where we replace $\mathrm{q}_1$ and $\mathrm{q}_2$ by $\mathrm{p}$, $\mathrm{q}$.

\end{proof}

In a central configuration (\ref{EQcc}) with multiplier $\lambda$ and total mass $M$ the quantity $\sqrt[3]{M/\lambda}$ is a special distance that we may call the intermediate distance.  With the normalization (\ref{EQs}) this distance is 1.

 \begin{figure}[H]
     \centering     \includegraphics[width=0.16\textwidth, angle=0]{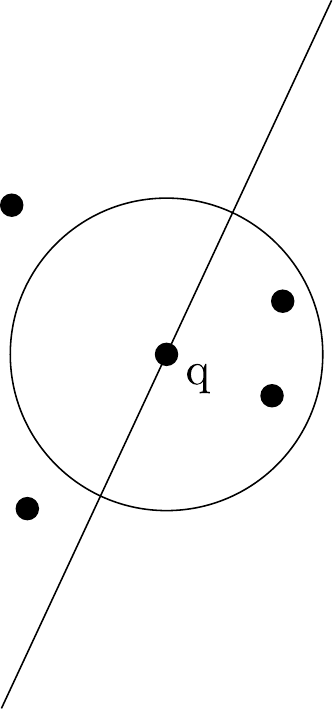}
     \caption{Normalized 5-body configuration which is not central whatever the masses.}
      \label{DST}
 \end{figure}

\begin{prop}\label{DSTtm}(Disk sector theorem)
For any planar central configuration with $\lambda=M$, for any body $\mathrm{q}$ of the configuration, consider the four open sectors delimited by an arbitrary line passing $\mathrm{q}$ and the unit circle centered at $\mathrm{q}$. Consider a domain formed by a half disk and the complementary of the other half disk in the other half plane. There are no bodies in this domain if and only if there are no bodies in the other similar domain (see Figure \ref{DST}).
\end{prop}

\begin{proof} Consider the sign of each term of the general equation (\ref{EQAg}) projected on the $\mathrm{y}$-axis.

\end{proof}

\begin{prop}\label{prop1}
   For a central configuration whose convex hull is a triangle, we number the points as in Figure \ref{tri}, that is, such that this triangle is $\mathrm{q}_1\mathrm{q}_2\mathrm{q}_3$ and such that $\mathrm{q}_1$, $\mathrm{q}_2$, $\mathrm{q}_5$, $\mathrm{q}_4$ in counterclockewise order form a convex configuration. Then it has the following properties:

\begin{enumerate}
    \item $r_{15}>r_{45}$\ and\  $r_{24}>r_{45}$;
    \item $r_{15}>r_{14}$\ and\  $r_{24}>r_{25}$;
    \item $r_{13}>r_{14}$\ and\  $r_{23}>r_{25}$.
\end{enumerate}
\end{prop}

 \begin{figure}[H]
     \centering     \includegraphics[width=0.18\textwidth, angle=0]{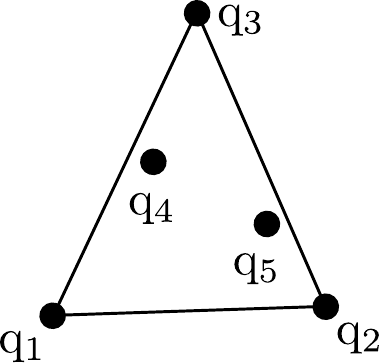}
     \caption{Central configuration whose convex hull is a triangle.}
      \label{tri}
 \end{figure}

\begin{proof}
\begin{enumerate}

  \item By renumbering the bodies, $r_{24}>r_{45}$ reduces to $r_{15}>r_{45}$. Assume the latter is not true. 
  Then we have $\mathrm{q}_5\in \mathrm{I}_{\mathrm{q}_4\mathrm{q}_1}$  (see Figure \ref{p1}).

 \begin{figure}[H]  
     \centering     \includegraphics[width=0.18\textwidth, angle=0]{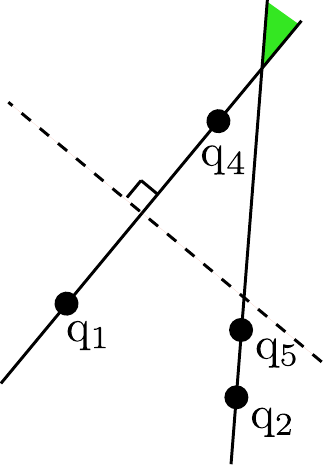}
     \caption{$\mathrm{q}_3$ is in the green region.}
     \label{p1}
 \end{figure}
  
By the hypothesis on the subconfiguration $\mathrm{q}_1$, $\mathrm{q}_2$, $\mathrm{q}_5$, $\mathrm{q}_4$, we have $\mathrm{q}_2\in \mathrm{I}^{\circ}_{\mathrm{q}_4\mathrm{q}_1}$. Since $\mathrm{q}_1$, $\mathrm{q}_2$, $\mathrm{q}_3$ are the 3 vertices of the convex hull of 5 points, we have $\mathrm{q}_3\in \mathrm{III}_{\mathrm{q}_4\mathrm{q}_1}$. This contradicts the perpendicular bisector theorem.
  
  \item By renumbering the bodies, $r_{24}>r_{25}$ reduces to $r_{15}>r_{14}$. Assume the latter is not true. Hence we have $\mathrm{q}_1\in \mathrm{IV}_{\mathrm{q}_4\mathrm{q}_5}$ (see Figure \ref{p2}).

\begin{figure}[H] 
     \centering     \includegraphics[width=0.2\textwidth, angle=0]{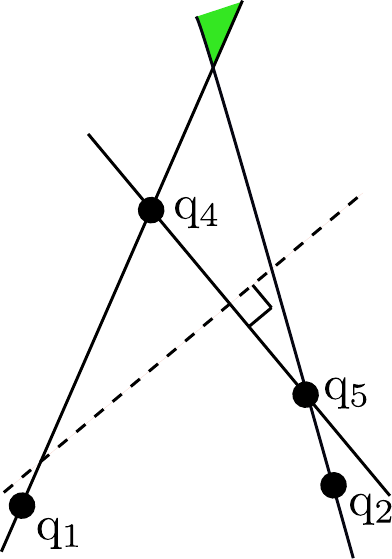}
     \caption{$\mathrm{q}_3$ is in the green region.}
     \label{p2}
 \end{figure}
  By the counterclockwise order of $\mathrm{q}_1$, $\mathrm{q}_2$, $\mathrm{q}_5$, $\mathrm{q}_4$, we must have $\mathrm{q}_2\in \mathrm{IV}_{\mathrm{q}_4\mathrm{q}_5}$. Since $\mathrm{q}_1$, $\mathrm{q}_2$, $\mathrm{q}_3$ form the vertices of the convex hull of 5 points, we have $\mathrm{q}_2\mathrm{q}_5\bigcap \mathrm{III}^{\circ}_{\mathrm{q}_4\mathrm{q}_5}=\emptyset$ and $\mathrm{q}_1\mathrm{q}_4\bigcap \mathrm{I}^{\circ}_{\mathrm{q}_4\mathrm{q}_5}=\emptyset$, hence $\mathrm{q}_3\notin \mathrm{III}^{\circ}_{\mathrm{q}_4\mathrm{q}_5}\bigcup \mathrm{I}^{\circ}_{\mathrm{q}_4\mathrm{q}_5}$. Therefore we have $\mathrm{q}_3\in \mathrm{II}_{\mathrm{q}_4\mathrm{q}_5}\bigcup \mathrm{IV}_{\mathrm{q}_4\mathrm{q}_5}$. This contradicts the perpendicular bisector theorem.

   \item By renumbering the bodies, $r_{23}>r_{25}$ reduces to $r_{13}>r_{14}$. Assume the latter is not true. Then $\mathrm{q}_1\in \mathrm{III}_{\mathrm{q}_3\mathrm{q}_4}$ (see Figure \ref{p3}).
  \begin{figure}[H]
   \centering    \includegraphics[width=0.36\textwidth, angle=0]{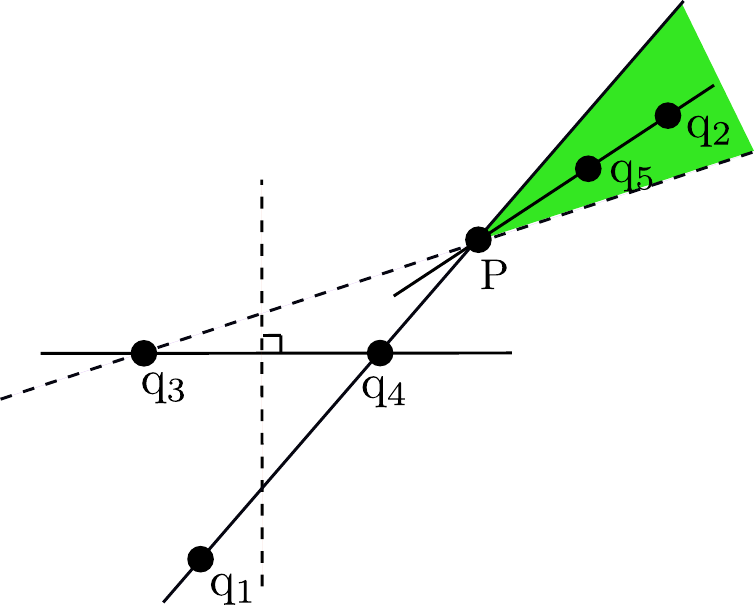}
     \caption{$\mathrm{q}_2$, $\mathrm{q}_5$ are in the green region.}
     \label{p3}
 \end{figure}
   Denote by $\mathrm{P}$ the intersection of ${\mathrm{q}_1\mathrm{q}_4}$ and ${\mathrm{q}_2\mathrm{q}_5}$. According to the hypothesis (see Figure \ref{tri}),  $\mathrm{q}_4$ is between $\mathrm{q}_1$ and $\mathrm{P}$.  Then we have $\mathrm{P}\in \mathrm{I}^{\circ}_{\mathrm{q}_3\mathrm{q}_4}$. Again by the hypothesis $\mathrm{q}_2$, $\mathrm{q}_5$ belong to the region delimited by ${\mathrm{q}_1\mathrm{q}_4}$ and $\mathrm{q}_3\mathrm{P}$ represented in Figure \ref{p3}. Then $\mathrm{q}_2$, $\mathrm{q}_5\in \mathrm{I}^{\circ}_{\mathrm{q}_3\mathrm{q}_4}$. This contradicts the perpendicular bisector theorem.
     
\end{enumerate}

\end{proof}

\begin{prop}\label{prop2or3}
    For the central configuration in Proposition\ \ref{prop1}, at least two of 
    $$\Ss_{12}<0\quad\hbox{or}\quad \Ss_{23}<0\quad {or}\quad \Ss_{13}<0,\quad\hbox{where}\quad \Ss_{ij}=\frac{1}{r_{ij}^3}-1,$$
    will happen.
\end{prop}

\begin{proof}
    Assume the conclusion is not true, then we have for example $\Ss_{13}\geq 0$ and $\Ss_{23}\geq 0$.
Choose a line passing through $\mathrm{q}_3$ such that all the bodies are on the same side of the line. According to Proposition \ref{DSTtm}, at least one of $\Ss_{34}\leq 0$ or $\Ss_{35}\leq 0$ shall happen. But this contradicts Proposition \ref{exter} applied to the distances from $\mathrm{q}_3$.

\end{proof}

 \begin{prop}\label{prop2}
For a central configuration whose convex hull is a quadrilateral, we number the points as in Figure \ref{quad}, that is, the vertices of the convex hull are $\mathrm{q}_1$, $\mathrm{q}_2$, $\mathrm{q}_3$, $\mathrm{q}_4$ counterclockwise. Then it has the following properties:
\begin{enumerate}
  \item $r_{13}>r_{15}$,  
  \item $r_{13}>r_{12}$ or $r_{13}>r_{14}$, 
\end{enumerate}
and the similar inequalities obtained by circular permutations of the indices 1234.
\end{prop}

 \begin{figure}[H]
     \centering     \includegraphics[width=0.2\textwidth, angle=0]{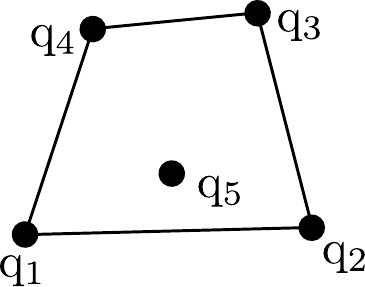}
     \caption{Central configuration whose convex hull is a quadrilateral.}
     \label{quad}
 \end{figure}

\begin{proof}
\begin{enumerate}
 \item Suppose that the inequality $r_{13}>r_{15}$ does not hold. Then $\mathrm{q}_1\in \mathrm{II}_{\mathrm{q}_3\mathrm{q}_5}$ or $\mathrm{q}_1\in \mathrm{III}_{\mathrm{q}_3\mathrm{q}_5}$. We may choose without loss of generality $\mathrm{q}_1\in \mathrm{II}_{\mathrm{q}_3\mathrm{q}_5}$. We also know that $\mathrm{q}_4\in \mathrm{III}_{\mathrm{q}_3\mathrm{q}_5}\bigcup \mathrm{IV}_{\mathrm{q}_3\mathrm{q}_5}$ and $\mathrm{q}_2\in \mathrm{I}_{\mathrm{q}_3\mathrm{q}_5}\bigcup \mathrm{II}_{\mathrm{q}_3\mathrm{q}_5}$, since $\mathrm{q}_5$ belongs to the convex sector delimited by the half-lines $\mathrm{q}_3\mathrm{q}_4$ and $\mathrm{q}_3\mathrm{q}_2$ and the order is counterclockwise.

We cannot have $\mathrm{q}_2\in \mathrm{I}_{\mathrm{q}_3\mathrm{q}_5}$, since $\mathrm{q}_5$ belongs to the convex sector delimited by the half-lines $\mathrm{q}_2\mathrm{q}_3$ and $\mathrm{q}_2\mathrm{q}_1$. The only remaining possibility is $\mathrm{q}_2\in \mathrm{II}_{\mathrm{q}_3\mathrm{q}_5}$. Then $\mathrm{q}_4\in \mathrm{IV}_{\mathrm{q}_3\mathrm{q}_5}$, otherwise, the quadrilateral will be in one side of line $\perp_{\mathrm{q}_3\mathrm{q}_5}$ which $\mathrm{q}_5$ is not in. But $\mathrm{q}_4\in \mathrm{IV}_{\mathrm{q}_3\mathrm{q}_5}$ contradicts the perpendicular bisector theorem.

\item Assume it is not true, then $r_{14}$, $r_{12}\geq r_{13}$ (see Figure \ref{p2.1}).
 \begin{figure}[H]
     \centering   \includegraphics[width=0.27\textwidth, angle=0]{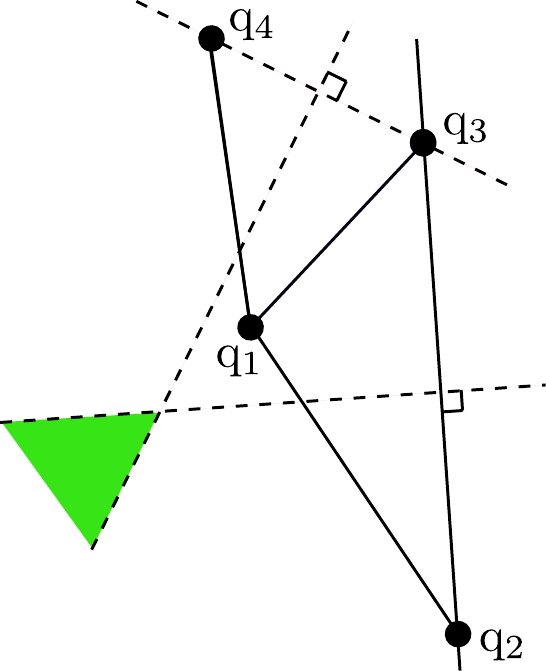}
     \caption{$\mathrm{q}_5$ is in the green region.}
     \label{p2.1}
 \end{figure}
 
 By $r_{14}\geq r_{13}$, $\perp_{\mathrm{q}_3\mathrm{q}_4}$ will pass $\overline{\mathrm{q}_1\mathrm{q}_4}$. By $r_{12}\geq r_{13}$,  $\perp_{\mathrm{q}_2\mathrm{q}_3}$ will pass $\overline{\mathrm{q}_1\mathrm{q}_2}$. So
$$\left ( \mathrm{I}_{\mathrm{q}_4\mathrm{q}_3}\bigcup \mathrm{III}_{\mathrm{q}_4\mathrm{q}_3}\right )\bigcap\left (\mathrm{II}_{\mathrm{q}_3\mathrm{q}_2}\bigcup \mathrm{IV}_{\mathrm{q}_3\mathrm{q}_2}\right )$$
is in the exterior of the convex hull of the 5 points. But $\mathrm{q}_5$ is in this intersection by the perpendicular bisector theorem. This contradicts the assumption that the convex hull of the 5 points is a quadrilateral with vertices $\mathrm{q}_1$, $\mathrm{q}_2$, $\mathrm{q}_3$, $\mathrm{q}_4$.

\end{enumerate}

\end{proof}

\begin{cor}\label{cor2}
    For the central configurations in Proposition\ \ref{prop2},
    $$\Ss_{13}<0\quad\hbox{and}\quad \Ss_{24}<0,\quad\hbox{where}\quad \Ss_{ij}=\frac{1}{r_{ij}^3}-1.$$
\end{cor}

\begin{proof}
    By renumbering the bodies, $\Ss_{24}<0$ reduces to $\Ss_{13}<0$. Assume the latter is not true. Together with $r_{13}>r_{15}$ in Proposition \ref{prop2} (1), this gives $\Ss_{15}> \Ss_{13}\geq 0$. Suppose for example that Proposition \ref{prop2} (2) gives $r_{13}>r_{12}$, which implies $\Ss_{12}> \Ss_{13}\geq 0$. Proposition \ref{DSTtm} applied to $\mathrm{q}_1$ and the line $\mathrm{q}_1\mathrm{q}_4$ then gives a contradiction.

\end{proof}

\bigskip

\section{Proof of Williams' inequality}

\bigskip

To prove Theorem \ref{mthm}, it is enough to prove what may be called Williams' inequality, that is
\begin{equation}\label{l}
\nu>0,
\end{equation}
where $\nu$ is the number which appears in identities (\ref{EQW}). According to (\ref{EQX}), the nonzero eigenvalues $\nu_1$, $\nu_2$ of $\Zs$ then satisfy $\nu_1\nu_2>0$. According to a mentioned result of \cite{Moe}, they also satisfy $\nu_1+\nu_2>0$ and are consequently positive. Replacing $\Zs$ by $\Z$, we get the statement of Theorem \ref{mthm}.

The expression $\Ss_{ik}\Ss_{jl}-\Ss_{il}\Ss_{jk}=\nu m_h\Delta_{ijh}\Delta_{klh}$ of identities (\ref{EQW}) and Williams' inequality imply that $\Ss_{ik}\Ss_{jl}-\Ss_{il}\Ss_{jk}$ and $\Delta_{ijh}\Delta_{klh}$ have same sign. More precisely, they are both positive, both negative or both zero. This is true whatever the permutation $12345\mapsto hijkl$. Reciprocally, if there is a single permutation $12345\mapsto hijkl$ such that $\Ss_{ik}\Ss_{jl}-\Ss_{il}\Ss_{jk}$ and $\Delta_{ijh}\Delta_{klh}$ are both positive, then (\ref{l}) is true.

For a strictly convex central configuration, we assume $\Delta_{123}>0$, then we have $\Delta_{ijk}>0$ for any $1\leq i<j<k\leq 5$. From Theorem \ref{cov}, we can deduce that (\ref{l}) is true. So to prove Theorem \ref{mthm}, there remains to prove it for the concave central configurations.

\begin{prop}\label{protri}
    Williams' inequality (\ref{l}) is true for any 5-body central configuration with triangular convex hull.
    
\end{prop}

\begin{cor}\label{lemtri}
    For a central configuration with a triangular convex hull, we number the points as in Figure \ref{tri}, that is, such that this triangle is $\mathrm{q}_1\mathrm{q}_2\mathrm{q}_3$ and such that $\mathrm{q}_1$, $\mathrm{q}_2$, $\mathrm{q}_5$, $\mathrm{q}_4$ in counterclockwise order form a convex configuration. The following inequalities hold:
    \begin{align}
    &\Ss_{23}\Ss_{45}\ \leq\ \Ss_{24}\Ss_{35}\ \leq\ \Ss_{25}\Ss_{34},\\
    &\Ss_{13}\Ss_{45}\ \leq\ \Ss_{15}\Ss_{34}\ \leq\ \Ss_{14}\Ss_{35},\\
    &\Ss_{12}\Ss_{45}\ \leq\ \Ss_{15}\Ss_{24}\ \leq\ \Ss_{25}\Ss_{14},\\
    &\Ss_{13}\Ss_{25}\  \leq\ \Ss_{12}\Ss_{35}\ \leq\ \Ss_{15}\Ss_{23},\\
    &\Ss_{23}\Ss_{14}\ \leq\ \Ss_{12}\Ss_{34}\ \leq\ \Ss_{13}\Ss_{24}.  
\end{align}
At least one of these inequalities is strict. Moreover, if any 3 of the 5 points are noncollinear, then all the inequalities are strict.
    
\end{cor}

 \begin{proof}
 
For the central configurations in Corollary \ref{lemtri}, we assume $\Delta_{123}>0$, then from Figure \ref{tri}, we have $\Delta_{124}$, $\Delta_{125}>0$, $\Delta_{234}$, $\Delta_{235}$, $\Delta_{345}$, $\Delta_{143}$, $\Delta_{153}$, $\Delta_{154}$, $\Delta_{254}\geq 0$. Then (\ref{EQW}) and (\ref{l}) imply the inequalities in Corollary \ref{lemtri}. Not all the $\Delta_{ijh}\Delta_{klh}$'s are zero, so (\ref{l}) implies that at least one inequality is strict. If all the triangular areas are nonzero, all the inequalities are strict.
 
 \end{proof}
       
\begin{proof}[Proof of Proposition \ref{protri}]
 We assume the contrary, namely $\nu\leq 0$. The corollary of $\nu\leq 0$ is similar to the corollary \ref{lemtri} of $\nu>0$, except that the inequalities are in reversed order and that we cannot deduce any strict inequality. We adopt the numbering convention of Corollary \ref{lemtri} and we get in particular the last three inequalities in reverse order: 
    \begin{align}
    &\Ss_{12}\Ss_{45}\  \geq \Ss_{15}\Ss_{24}\  \geq \Ss_{25}\Ss_{14}, \label{f5}\\
    &\Ss_{13}\Ss_{25}\  \geq \Ss_{12}\Ss_{35}\  \geq \Ss_{15}\Ss_{23}, \label{f4}\\
    &\Ss_{23}\Ss_{14}\  \geq \Ss_{12}\Ss_{34}\  \geq \Ss_{13}\Ss_{24}. \label{f3}
\end{align}
By Proposition \ref{prop1} $(2)\ (3)$,  $r_{15}>r_{14}$ and  $r_{13}>r_{14}$,  so $\Ss_{15}<\Ss_{14}$ and  $\Ss_{13}<\Ss_{14}$. By Proposition \ref{DSTtm} applied to $\mathrm{q}_1$ and the line $\mathrm{q}_1\mathrm{q}_2$, we have $\Ss_{14}>0$. Similarly, $\Ss_{24}<\Ss_{25}$,  $\Ss_{23}<\Ss_{25}$ and $\Ss_{25} > 0$.  The inequality $\Ss_{25}\Ss_{14}\ \leq\ \Ss_{15}\Ss_{24}$ in (\ref{f5}) is incompatible with $\Ss_{15}<\Ss_{14}$ and $\Ss_{24}<\Ss_{25}$ except maybe if $\Ss_{15}<0$ and $\Ss_{24}<0$. According to Proposition \ref{prop2or3}, $\Ss_{23}<0$ or $\Ss_{13}<0$. Suppose for example $\Ss_{23}<0$. Then (\ref{f4}) gives $\Ss_{13}>0$. Proposition \ref{prop2or3} now gives $\Ss_{12}<0$ and we have the signs of all the $\Ss_{1j}$'s. We apply Proposition \ref{DSTtm} to $\mathrm{q}_1$ and for example the line $\mathrm{q}_1\mathrm{q}_4$ and get a contradiction. So $\nu>0$.

\end{proof}

\begin{prop}\label{proquad}
    Williams' inequality (\ref{l}) is true for any 5-body central configuration with quadrilateral convex hull.
    
\end{prop}

\begin{cor}\label{lemquad}
For a central configuration with a quadrilateral convex hull, we number counterclockwise the points at the vertices of the convex hull as $\mathrm{q}_1$, $\mathrm{q}_2$, $\mathrm{q}_3$, $\mathrm{q}_4$ (see Figure \ref{quad}). The following inequalities hold:
\begin{align}
    &\Ss_{23}\Ss_{45}\ \geq\ \Ss_{35}\Ss_{24}\ \leq\ \Ss_{25}\Ss_{34},\\
    &\Ss_{15}\Ss_{34}\ \geq\ \Ss_{13}\Ss_{45}\ \leq\ \Ss_{14}\Ss_{35},\\
    &\Ss_{14}\Ss_{25}\ \geq\ \Ss_{15}\Ss_{24}\ \leq\ \Ss_{12}\Ss_{45},\\
    &\Ss_{12}\Ss_{35}\ \geq\ \Ss_{13}\Ss_{25}\ \leq\ \Ss_{15}\Ss_{23},\\
    &\Ss_{12}\Ss_{34}\ \leq\ \Ss_{13}\Ss_{24}\ \geq\ \Ss_{14}\Ss_{23}.
\end{align}
At least one of these inequalities is strict. Moreover, if any 3 of the 5 points are noncollinear, then all the inequalities are strict.
\end{cor}

\begin{proof} 
For the central configurations in Corollary \ref{lemquad}, we assume $\Delta_{123}>0$, then from Figure \ref{quad}, we have $\Delta_{124}$, $\Delta_{134}$, $\Delta_{234}>0$, $\Delta_{125}$, $\Delta_{235}$, $\Delta_{154}$, $\Delta_{345}\geq 0$. Here we accept $\mathrm{q}_5$ on a diagonal of the quadrilateral, so we do not decide the signs of $\Delta_{135}$ and $\Delta_{245}$. Williams' inequality (\ref{l}) and identities (\ref{EQW}) imply the inequalities in Corollary \ref{lemquad}. Not all the $\Delta_{ijh}\Delta_{klh}$'s are zero, so (\ref{l}) implies that at least one inequality is strict. If all the triangular areas are nonzero, all the inequalities are strict (and indeed we deduce another five strict inequalities from the signs of $\Delta_{135}$ and $\Delta_{245}$).
\end{proof} 

\begin{proof}[Proof of Proposition \ref{proquad}] Again, we assume $\nu\leq 0$, adopt the numbering convention of Corollary \ref{lemquad} and deduce the same inequalities in reverse order. The fifth one in reverse order is
\begin{equation}\label{g6}
\Ss_{13}\Ss_{24}\leq\Ss_{12}\Ss_{34}
\end{equation}
and
\begin{equation}\label{g7}
\Ss_{13}\Ss_{24}\leq\Ss_{14}\Ss_{23}.
\end{equation}
According to Corollary \ref{cor2}, $\Ss_{13}\Ss_{24}>0$ since both factors are negative. Consider the two shortest exterior sides.  We have two cases: either they are adjacent or they are opposite.

In the adjacent case, the shortest sides are for example $\overline{\mathrm{q}_1\mathrm{q}_2}$ and $\overline{\mathrm{q}_2\mathrm{q}_3}$. According to  Proposition \ref{prop2} (2), $r_{13}>r_{12}$ and $r_{13}>r_{23}$. The same proposition also gives $r_{24}>r_{14}$ or $r_{24}>r_{34}$. We can assume without loss of generality that the third shortest side is $\overline{\mathrm{q}_1\mathrm{q}_4}$. Then $r_{24}>r_{14}$, $\Ss_{13}<\Ss_{23}$, $\Ss_{24}<\Ss_{14}$ and (\ref{g7}) cannot be true if $\Ss_{23}<0$ and $\Ss_{14}<0$. As their product is positive, $\Ss_{23}>0$ and $\Ss_{14}>0$. Now, $\Ss_{12}>0$ since $\overline{\mathrm{q}_1\mathrm{q}_2}$ is shorter than $\overline{\mathrm{q}_1\mathrm{q}_4}$. Finally, $\Ss_{34}>0$ by (\ref{g6}).

In the opposite case, the shortest sides are for example $\overline{\mathrm{q}_1\mathrm{q}_2}$ and $\overline{\mathrm{q}_3\mathrm{q}_4}$.  According to  Proposition \ref{prop2} (2), $r_{13}>r_{12}$, $r_{13}>r_{34}$, $r_{24}>r_{12}$ and $r_{24}>r_{34}$. Consequently (\ref{g6}) is impossible except maybe if $\Ss_{12}>0$ and $\Ss_{34}>0$.

So in both cases we conclude that two opposite exterior sides, for example $\overline{\mathrm{q}_1\mathrm{q}_2}$ and $\overline{\mathrm{q}_3\mathrm{q}_4}$, are ``short'', in the sense that $\Ss_{12}>0$ and $\Ss_{34}>0$. But the fourth inequality, left, in reverse order is
\begin{equation}\label{g8}
\Ss_{12}\Ss_{35}\leq \Ss_{13}\Ss_{25}.
\end{equation}
From it we deduce that the $\Ss_{i5}$, $i=1,\dots,4$, cannot be all positive. Indeed, if $\Ss_{35}>0$, then $\Ss_{25}<0$.

Let us first suppose that $\Ss_{15}\leq 0$. Then by Proposition \ref{DSTtm} applied to the point $\mathrm{q}_1$ and the line $\mathrm{q}_1\mathrm{q}_2$, we have $\Ss_{14}>0$. Then (\ref{g7}) gives $\Ss_{23}>0$ and all the exterior sides are ``short''. But $\Ss_{15}\leq 0$ so $\overline{\mathrm{q}_1\mathrm{q}_5}$ is not short. Proposition \ref{exter} tells that $\mathrm{q}_5$ is strictly exterior to the triangle $\mathrm{q}_1\mathrm{q}_2\mathrm{q}_4$.

If we had moreover $\Ss_{35}\leq 0$, $\mathrm{q}_5$ would also be strictly exterior to the triangle $\mathrm{q}_2\mathrm{q}_3\mathrm{q}_4$, which is impossible. So $\Ss_{35}>0$. But (\ref{g8}) then gives $\Ss_{25}<0$, while $\Ss_{14}\Ss_{35}\leq \Ss_{13}\Ss_{45}$, which is the second inequality, right, in reverse order, gives $\Ss_{45}<0$. But $\Ss_{25}<0$ and $\Ss_{45}<0$ give the same contradiction as $\Ss_{15}\leq 0$ and $\Ss_{35}\leq 0$.

 If instead of $\Ss_{15}\leq 0$ we suppose that $\Ss_{25}\leq 0$, $\Ss_{35}\leq 0$ or $\Ss_{45}\leq 0$, the reasoning follows the same steps with some changes in the indices. If for example $\Ss_{25}<0$, we exchange indices 1 and 2 and indices 3 and 4. So, we use, in reverse order, the third inequality, right instead of the fourth, left, etc. We get the same contradiction in all the cases. So $\nu>0$.

\end{proof}

{\it Acknowledgements.} We wish to thank Christian Marchal and Kuo-Chang Chen for sharing important ideas and interesting information.
Jiexin Sun acknowledges the support of the China Scholarship Council program (Project ID: 202306220174).

\end{document}